\documentclass[a4paper,UKenglish,cleveref, autoref, thm-restate]{lipics-v2021}
\usepackage{mathtools}
\DeclarePairedDelimiter{\ceil}{\lceil}{\rceil}
\usepackage[colorinlistoftodos,prependcaption,textsize=small]{todonotes}
\DeclareMathOperator*{\argmax}{arg\,max}
\DeclareMathOperator*{\argmin}{arg\,min}
\nolinenumbers
\newcommand{\old}[1]{{}}

\title{Optimizing Visibility-Based Search in Polygonal Domains} %TODO Please add

\author{Kien C. Huynh}{Linköping University, Sweden\footnote{Work done at Stony Brook University}}{kchuynh@cs.stonybrook.edu}{https://orcid.org/0000-0001-6247-8964}{}%TODO mandatory, please use full name; only 1 author per \author macro; first two parameters are mandatory, other parameters can be empty. Please provide at least the name of the affiliation and the country. The full address is optional. Use additional curly braces to indicate the correct name splitting when the last name consists of multiple name parts.

\author{Joseph S. B. Mitchell}{Stony Brook University, USA}{joseph.mitchell@stonybrook.edu}{https://orcid.org/0000-0002-0152-2279}{}

\author{Linh Nguyen\footnote{Corresponding author}}{Stony Brook University, USA}{linh.nguyen.1@stonybrook.edu}{https://orcid.org/0009-0009-3518-929X}{}

\author{Valentin Polishchuk}{Linköping University, Sweden}{valentin.polishchuk@liu.se}{https://orcid.org/0000-0002-8292-2281}{}

\authorrunning{Kien C. Huynh, Joseph S. B. Mitchell, Linh Nguyen, Valentin Polishchuk} %TODO mandatory. First: Use abbreviated first/middle names. Second (only in severe cases): Use first author plus 'et al.'
\Copyright{Authors} %TODO mandatory, please use full first names. LIPIcs license is "CC-BY";  http://creativecommons.org/licenses/by/3.0/

 %TODO mandatory: Please choose ACM 2012 classifications from https://dl.acm.org/ccs/ccs_flat.cfm 

 %TODO mandatory; please add comma-separated list of keywords

\category{} %optional, e.g. invited paper

\relatedversion{} %optional, e.g. full version hosted on arXiv, HAL, or other respository/website
\acknowledgements{This work is partially supported by the National Science Foundation (CCF-2007275) and by the Swedish Research Council and the Swedish Transport Administration}%optional

\EventNoEds{2}
\EventLongTitle{19th Scandinavian Symposium on Algorithm Theory (SWAT 2024)}
\EventShortTitle{SWAT 2024}
\EventAcronym{SWAT}
\EventYear{2024}
\EventDate{June 12--14, 2024}
\EventLocation{Helsinki, Finland}
\EventLogo{}
\ArticleNo{36}
\begin{document}

\maketitle

\begin{abstract}
Given a geometric domain $P$, visibility-based search problems seek routes for one or more mobile agents (``watchmen'') to move within $P$ in order to be able to see a portion (or all) of $P$, while optimizing objectives, such as the length(s) of the route(s), the size (e.g., area or volume) of the portion seen, the probability of detecting a target distributed within $P$ according to a prior distribution, etc. The classic watchman route problem seeks a shortest route for an observer, with omnidirectional vision, to see all of $P$. In this paper we study bicriteria optimization problems for a single mobile agent within a polygonal domain $P$ in the plane, with the criteria of route length and area seen. Specifically, we address the problem of computing a minimum length route that sees at least a specified area of $P$ (minimum length, for a given area quota). We also study the problem of computing a length-constrained route that sees as much area as possible. We provide hardness results and approximation algorithms. In particular, for a simple polygon $P$ we provide the first fully polynomial-time approximation scheme for the problem of computing a shortest route seeing an area quota, as well as a (slightly more efficient) polynomial dual approximation. We also consider polygonal domains $P$ (with holes) and the special case of a planar domain consisting of a union of lines. Our results yield the first approximation algorithms for computing a time-optimal search route in $P$ to guarantee some specified probability of detection of a static target within $P$, randomly distributed in $P$ according to a given prior distribution.
\end{abstract}

\section{Introduction}
We investigate the \textsc{Quota Watchman Route} problem (QWRP) and the \textsc{Budgeted Watchman Route} problem (BWRP) for a single mobile agent (a ``watchman'') within a polygonal domain $P$ in the plane. These problems naturally arise in various applications, including motion planning, search-and-rescue, surveillance, and exploration of a polygonal domain, where complete coverage is not feasible due to shortage of fuel, time, etc. The QWRP seeks a route/tour that sees at least some specified area of the domain $P$ with a shortest length, while the BWRP seeks a route/tour that sees the maximum area subject to a length constraint. Both can be seen as extensions of the well-known \textsc{Watchman Route Problem} (WRP) with different objectives and constraints.

The challenge in addressing the trade-off between area seen and tour length is that one is not able to exploit the optimality structure that is implied by having to see \emph{all} of a polygon $P$. It is this structure, yielding an ordered sequence of ``essential cuts'', that allows the WRP to be solved efficiently, e.g., as an instance of the ``touring polygons problem'' \cite{dror2003touring}.

{\bf Results:} We address the challenge by establishing new structural results that enable a careful discretization and analysis, along with carefully crafted dynamic programs. We provide several new results on optimal visibility search in a polygon:
\begin{description}
    \item[(1)] We prove that the QWRP and the BWRP are (weakly) NP-hard, even in a simple polygon; this is to be contrasted with the WRP, for which exact polynomial-time algorithms are known in simple polygons. 
    \item[(2)] For the QWRP in a simple polygon $P$, we give the first fully polynomial-time approximation scheme (FPTAS), as well as a dual-approximation (with slightly more efficient running time than the FPTAS) that computes a tour having length at most $(1+\varepsilon_1)$ times the length of an optimal tour that sees area at least $A$ (where $A$ is the area quota), while seeing area at least $(1-\varepsilon_2)A$ for any $\varepsilon_1, \varepsilon_2 > 0$. 
    \item[(3)] For the BWRP in simple $P$, we compute, in polynomial time,  a tour of length at most $(1 + \varepsilon)B$ that sees a region within $P$ of area at least that seen by an optimal tour of length at most $B$. 
    \item[(4)] In a multiply connected domain, in a polygon $P$ with holes, we provide hardness of approximations and a $(1 + \varepsilon, O(\log n))$-dual approximation ($n$ is the number of vertices of $P$) for the BWRP\@. In the special case of an arrangement of lines, we obtain polynomial-time exact algorithms for both problems.

    \old{SoCG Reviewer 2: line 47:  log(OPT) does not make sense in this form, as it is not scale-invariant.    The approximation factor doesn't improve when you scale the input by, say, 0.001.}

    \item[(5)] We solve two visibility-based stochastic search problems that seek to locate a static target given a prior probability distribution of its location within $P$: (a) compute the minimum time to achieve a specified detection probability; (b) compute a search route maximizing the probability of detection by time $T$ for a mobile searcher.
    \end{description}
    
\subsection*{Related Work}

Chin and Ntafos introduced the classic \textsc{Watchman Route Problem} (WRP) \cite{first}: compute a shortest closed route (tour) within a polygon $P$ from which every point of $P$ can be seen; they gave an $O(n)$-time algorithm for computing an optimal tour in a simple orthogonal polygon, and later results established polynomial-time exact algorithms for the WRP in a simple polygon $P$, both with and without an anchor point (depot)~\cite{carlsson1999finding, first,chin1991shortest,dror2003touring,ntafos1994optimum,tan2001fast, corrigendum}. In a polygon $P$ with holes, the WRP is NP-hard~\cite{first,dumitrescu2012watchman} and is, in fact, NP-hard to approximate better than a logarithmic factor~\cite{mitchell2013approximating}; however, an $O(\log^2 n)$-approximation algorithm is known~\cite{mitchell2013approximating}. The BWRP and the QWRP are natural variants of the WRP.   
%\textsc{Watchman Route}. The \textsc{Watchman Route} problem in a domain with holes was shown to be NP-hard in \cite{first}, however the proof was flawed and later corrected in \cite{dumitrescu2012watchman}. The best running times in a simple domain without holes are $O(n^3\log n)$ when a starting point, i.e. a point on the boundary that the route passes through is given (we call this version \textit{anchored} WRP), and $O(n^4\log n)$ when no such point is given (\textit{floating} WRP). Efficient approximation algorithms with constant factors are also known, see \cite{tan2004approximation, tan2007linear}. In polygons with holes, Mitchell proved that unless P = NP, the problem cannot be approximated to within a factor of $c\log n$ for some constant $c$ (where $n$ is the number of vertices of $P$) and gave an $O(\log^2n)-$approximation \cite{mitchell2013approximating}. 

Another related problem is that of maximum visibility coverage with point guards: Given an integer $k$, place $k$ point guards within $P$ to maximize the area of $P$ seen by the guards. When $k$ is arbitrary, the problem is NP-hard \cite{ntafos1994optimum}, since an exact solution to this problem would yield a method to compute the minimum number of guards needed to see a polygon. Viewed as a maximum coverage problem, one can greedily compute an approximation, with factor $\left(1 - \frac{1}{e}\right)$, by iteratively placing a guard that sees the most unseen area \cite{cheong2007finding, ntafos1994optimum}.

The BWRP is related to the \textsc{Orienteering} problem. Given a budget constraint and an edge-weighted graph where each vertex is associated with a prize, the objective of \textsc{Orienteering} is to find a path/tour within the length budget maximizing the total reward of the vertices visited. On the other hand, the QWRP is related to the \textsc{Quota Traveling Salesperson} problem, which aims to minimize the distance travelled to achieve a given quota of reward. The Euclidean versions of \textsc{Orienteering} and \textssc{Quota} TSP have polynomial-time approximation schemes~\cite{chen2008euclidean,gottlieb2022faster,mitchell2000geometric}. Both the QWRP and the BWRP can be considered a reward (the area of $P$ seen by the watchman) collecting process; however, the main difference lies in the continuous nature of visibility, since we see portions of the domain as we travel to checkpoints, we must take into account the area that has been seen previously.

Optimal search theory has been extensively studied in discrete, graph theoretic settings; see, e.g., \cite{eagle1984optimal, eagle1990optimal, trummel1986complexity}. In geometric contexts, searching and target tracking have been studied in the form of \textsc{Visibility-based Pursuit-Evasion} games. In \cite{guibas1997visibility}, the visibility-based version of the pursuit-evasion game was introduced and formulated as a geometric problem, in which an evader moves unpredictably, arbitrarily fast within a polygonal domain, and the goal is to strategically coordinate one or multiple pursuers to guarantee a finite time of detection. See the survey~\cite{chung2011search} on visibility-based pursuit-evasion games.

\section{The QWRP in a Simple Polygon}
\label{sec:quota}
\subsection{Preliminaries and Hardness Results}
A \textit{simple polygon} $P$ is a simply connected subset of $\mathbb{R}^2$ whose boundary, $\partial P$, is a polygonal cycle consisting of a finite set of line segments, whose endpoints are the \emph{vertices}, $v_1,v_2,\ldots,v_n$, of $P$. A vertex is \textit{reflex} (resp. \textit{convex}) if its internal angle is at least (resp. at most) 180 degrees. We consider polygons to be closed sets, including the interior and the boundary. We use the notation $|\cdot|$ to denote the measure of several types of objects. In the case of a segment or a route $\gamma$, $|\gamma|$ denotes its length, while for a polygon $P$, $|P|$ denotes the area of $P$. For a finite set $S$, $|S|$ is the cardinality of $S$. Based on the object within the notation, the interpretation should be apparent.

Point $x\in P$ sees point $y\in P$ if the line segment connecting them lies entirely within $P$. The \textit{visibility polygon} of $x$, denoted $V(x)$, is the closed region of $P$ that $x$ sees; necessarily, $V(x)$ is a simple polygon within $P$. For a subset $X\subset P$, let $V(X)$ be the set of points that are seen by at least one point in $X$; formally, $V(X) = \bigcup_{x\in X}V(x)$. The visibility polygon of a point or a segment can be computed in time $O(n)$ for a simple polygon, or in time $O(n+h\log h)$ for a polygonal domain with $n$ vertices and $h$ holes \cite{o2017visibility}. Given a domain $P$ (a simple polygon or a polygon with holes) and an area quota $0 \le A\le |P|$, in the QWRP, the objective is to find a \textit{tour} (a polygonal cycle) $\gamma \subset P$ of minimum length $|\gamma|$ such that $|V(\gamma)| \ge A$; see Figure~\ref{fig:relativeconvex}. Note that when $A = |P|$, QWRP is the classic \textsc{Watchman Route Problem}.

\begin{figure}[h]
        \centering
        \includegraphics[width=0.8\textwidth]{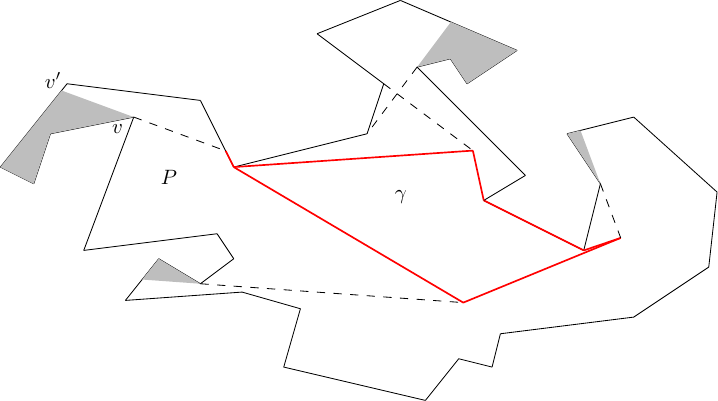}
        \caption{A route $\gamma$ (red) that sees the white portion of $P$ (the gray regions are unseen).}
        \label{fig:relativeconvex}
    \end{figure}
We also distinguish between the \textit{anchored} version (in which $\gamma$ must pass through a given depot point $s$) and the \textit{floating} version (in which no depot is given). 
% We present the algorithm to solve the anchored BWRP, where a starting point is given, as is the case with many practical applications. 
We provide the following NP-hardness results for both the anchored and floating cases: 

\begin{theorem}\label{thm:weakly-np-hard-qwrp}
The QWRP in a simple polygon is weakly NP-hard, with or without an anchor.  % point $s$.
\end{theorem}
\begin{proof}
    We use a reduction from \textsc{Inverse Knapsack}, which is known to be NP-hard, to show that QWRP is NP-hard, even in a simple polygon. Consider an instance of \textsc{Inverse Knapsack}: Given a value quota $V$, items of integral weights $(l_1,\ldots,l_m)$, and values $(V_1,\ldots,V_m)$, we are to determine a combination of items having minimum total weight such that their collective value is at least~$V$.

We construct a polygon similar to what is shown on the right of Figure~\ref{fig:floating-np-hard-gadget}. The construction is built from several ``roof'' gadgets (see Figure~\ref{fig:floating-np-hard-gadget}, left).
Each gadget consists of a rectangle of side lengths $\delta$ and $3\delta$, where $\delta \ll \frac{1}{1000m}\min\{l_i\}$ and two wings of area $V_i$ in total.
To see both wings, the watchman needs to visit a point in the middle of the rectangle. If the watchman chooses to visit only the bottom of the rectangle instead, there will be a savings of at most $\delta$ in distance travelled.
Additionally, there are roofs with large area $V_{\max} \gg \sum_{i=1}^m{V_i}$.
Along the boundary of this polygon are $m$ roof gadgets of areas $V_1,\ldots V_m$, each neighbors two gadgets of area $V_{\max}$.
We let $V_c$ denote the area of the remaining portion of the polygon without the $V_{\max}$ gadgets and the wings of the $V_i$ gadgets.
Traveling from one gadget of area $V_{\max}$ to another of the same area requires a cost of about $l_0 \gg B$.
If the watchman detours to visit a rectangle of some $V_i$ roof, he must travel an additional distance of at least $l_0 + l_i - \delta$.
We give the watchman a quota of $V_c + mV_{\max} + V$. Since the $V_{\max}$ gadgets provide very large profits, the tour has to visit all of them, as well as a subset of the remaining roofs whose total area is at least $V$. Note that the tour does not necessarily commit to seeing all gadgets chosen completely. However, by design, the difference in distance travelled between seeing the gadgets ``completely'' and ``partially'' is at most $m\delta \ll \frac{1}{1000}\min\{l_i\}$ and thus, the watchman must necessarily select a subset of roofs that minimizes the total detour length.
\begin{center}
    \begin{figure}[h]
        \centering
        \includegraphics[width=\textwidth]{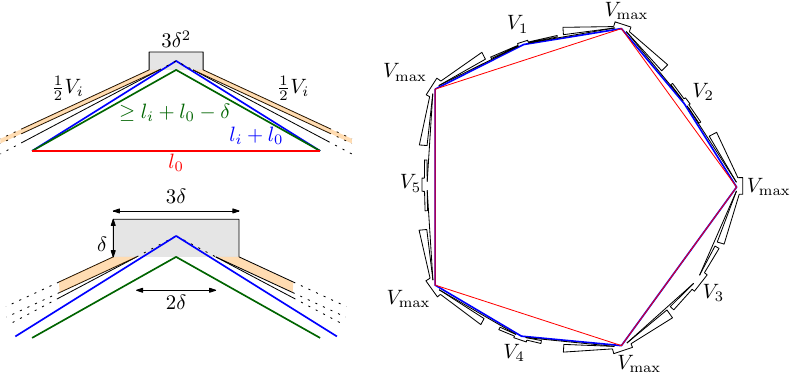}
        \caption{Left: The roof gadgets (the $V_i$ gadget above and the $V_{\max}$ gadget below) used in the hardness proof for the QWRP and the BWRP. 
        Note that the illustration is not drawn to scale.
        Right: The red tour has length $ml_0$; it is a shortest tour visiting all of the $V_{\max}$ gadgets. 
        By making a detour of length at least $l_i - \delta$, an additional area $V_i$ can be seen. The detour choices are independent of each other, due to the $V_{\max}$ gadgets that isolate them.}
        \label{fig:floating-np-hard-gadget}
    \end{figure}
\end{center}
Given a tour that sees $V_c + mV_{\max} + V$, we can convert it to a solution for the \textsc{Inverse Knapsack} instance with profit at least $V$ by looking at the set of roofs the watchman visits other than the ones with $V_{\max}$ area.
Conversely, given an \textsc{Inverse Knapsack} solution that has profit $V$, we can create a tour that sees area at least $V_c + mV_{\max} + V$.
The above construction is for the floating version of QWRP;
to show that the anchored version is also NP-hard, we simply use the same gadgets and place the starting point in the gap of any $V_{\max}$ roof.

\end{proof}
Throughout the paper, we assume a real RAM model of computation~\cite{shamos1978computational}.

%% longer form, for fuller paper:
\old{It is crucial to highlight that due to the irrational nature of Euclidean lengths, we assume a real RAM model of computation, one which allows for addition, subtraction and comparison of real numbers (in constant time per operation) \cite{shamos1978computational}. The LEDA library \cite{naher1990leda} (among others) mimics real RAM using specialized data structures, allowing running time analysis based on the number of library calls.
}

\subsection{Structural Lemma}
\label{subsec:struct_lemmas}

 Let $\pi_P(x,y)$ denote the \textit{geodesic shortest path} (shortest path constrained to stay within $P$) between $x\in P$ and $y \in P$; $\pi_P(x,y)$ is unique in a simple $P$, and is the segment $xy$ if $x$ sees $y$. For a subset $S\subseteq P$, the \textit{relative/geodesic convex hull} of $S$ is the minimal set that contains $S$ and is closed under taking shortest paths. Equivalently, the relative convex hull of $S$ is the minimum-perimeter connected subset of $P$ that contains $S$. A set is \textit{relatively convex} if it is equal to its relative convex hull, and a closed curve is relatively convex if it is the boundary of a relatively convex set. Let $P_\gamma$ denote the connected region bounded by some closed polygonal chain $\gamma$. If $P_\gamma$ is a (sub)polygon of $P$ and $P_\gamma$ is relatively convex, then $P_\gamma$ is the relative convex hull of its convex vertices, and all reflex vertices of $P_\gamma$ are necessarily reflex vertices of $P$. We similarly define relative convexity of an open polygonal chain $\gamma$: if $\gamma$ is a connected subset of the boundary of the relative convex hull of $\gamma$, then we say that $\gamma$ is relatively convex. Geodesic shortest paths and relative convex hulls have been studied extensively and can be computed efficiently \cite{mitchell2000geometric}. 

An optimal solution to the QWRP in a simple polygon $P$ satisfies a structural lemma:

\begin{comment}
JOE: We need to restate slightly for the case of a depot that is forced to be on the route, since then the route need not be relatively convex, of course.

JOE: We may want to think about PATH versions of our problem too?? We consider cycles (with and without a depot) already. Watchman path problems tend to be tricky, but doable with DP.  There is also the TREE version: find a tree of length at most $L$ that sees the most.
\end{comment}
\begin{comment}
        PRIOR LEMMA: Given a simple polygon $P$ and a length budget $L$, there exists an optimal budgeted watchman route subject to a length budget $L$ that is a polygonal tour relatively convex to $P$.
    REWORDED (ANY optimal route is relatively convex, right?): and combined with the other facts into one lemma
\end{comment}
\begin{lemma}
\label{lem:structure}
    For a simple polygon $P$ with $n$ vertices, and no depot, an optimal QWRP tour is a relatively convex simple polygonal cycle of at most $2n$ vertices. 
\end{lemma}

\begin{proof}
Let $\gamma$ be an optimal QWRP tour and let $P' = V(\gamma)$ be the visibility polygon of $\gamma$. Since $\gamma$ is connected, $P'$ is a simple subpolygon of $P$; some edges of $P'$ coincide with edges of $P$ and some are shadow chords (chords separating $V(\gamma)$ from the rest of $P)$ supported by reflex vertices of $P$. Then $\gamma$ is a shortest watchman route in the simple polygon $P'$. Thus, $\gamma$ is relatively convex in $P'$, and thus in $P$, and $\gamma$ has at most $2n$ vertices, since $P'$ is easily seen to have at most $n$ vertices. (See \cite{chin1991shortest,mitchell2013approximating}.)
    
    Specifically, the polygon $P'=V(\gamma)$ is obtained from $P$ by removing certain subpolygons (``shadow pockets'') of $P$ that are each defined by a chord, $vv'$, extending from a reflex vertex, $v$, of $P$, along the line through $v$ and a convex vertex of $\gamma$, to the first point $v'$ on the boundary of $P$. This process introduces a (convex) vertex $v'$ (on an edge of $P$, in general on its interior), and removes at least one vertex of $P$, on the boundary of the pocket that is cut off by the chord. Refer to Figure~\ref{fig:relativeconvex}. Thus, $P'$ has at most $n$ vertices.
    For a simple polygon with $n$ vertices, any shortest watchman route has at most $2n$ vertices \cite{carlsson1999finding, first,chin1991shortest,dror2003touring,ntafos1994optimum,tan2001fast, corrigendum}. Moreover, all reflex vertices of $P'$ must be reflex vertices of $P$, hence all reflex vertices of $\gamma$ must also be reflex vertices of $P$.
\end{proof}

    If there is a specified depot $s\in\partial P$, a statement similar to Lemma~\ref{lem:structure} holds. If $s$ is interior to $P$, an optimal tour $\gamma=(s,w_1,w_2,\ldots,w_k,s)$ through $s$ need not be relatively convex; however, it is ``nearly'' relatively convex in that the tour obtained by replacing the two edges $sw_1$ and $w_ks$ with the geodesic path $\pi_P(w_1,w_k)$ is relatively convex. 
    % xxx in full paper we may want to give a proof
    
    \subsection{Dual approximation algorithm for anchored QWRP} 
    \label{sec:dualapproxqwr}
    An optimal tour for the QWRP will, in general, have (convex) vertices that are interior to $P$, at locations within the continuum that are not known to come from a discrete set. This poses a challenge to algorithms that are to compute solutions for the QWRP exactly or approximately. We address this challenge by discretizing an appropriate portion of the domain $P$ using a (Steiner) triangulation whose faces are small enough that we can afford to round an optimal tour to vertices of the triangles, while increasing the length of the tour only slightly, and assuring that the rounded tour continues to see at least as much of $P$ as the optimal tour did. We focus here on the anchored case, with a specified depot $s$, which we assume to be on $\partial P$ for now.

    First, we triangulate $P$ (in $O(n)$ time \cite{chazelle1991triangulating}), including $s$ as a vertex of the triangulation. We then overlay, centered on $s$, a regular square grid of pixels of side lengths $\delta$ within an axis-aligned square of size $L$-by-$L$ for a length $L$ that is at least the optimal tour length; we specify how to determine $\delta$ and $L$ below. The overlay of the grid with the triangulation yields a partition of $P$ into convex cells of constant complexity, each of which we triangulate, resulting in an overall Steiner triangulation of $P$, such that every triangle within distance $L/2$ of $s$ has diameter at most $\sqrt{2}\delta$ and perimeter at most $4\delta$; we let $S_{\delta,L}$ denote the set of vertices of these triangles. We refer to $S_{\delta,L}$ as the set of \textit{candidate turn points} for a route.

    \begin{lemma}
    \label{lem:snapping} For an optimal tour $\gamma$ for the QWRP with area quota $A$, there exists a polygonal tour $\gamma'$ whose vertices are in the set $S_{\delta,L}$ of candidates, such that $\gamma'$ is relatively convex, $|\gamma'|\le |\gamma|+(8+4\sqrt{2})\delta n$ and $V(\gamma) \subseteq V(\gamma')$.  
    \end{lemma}
    \begin{proof}
Let $c_1,c_2,\dots$ be convex vertices of the optimal tour $\gamma$ ($\gamma$ is the relative convex hull of such vertices) and let $\sigma_1,\sigma_2,\ldots$ be (closed) cells of the decomposition that contain the vertices. Let $\gamma'$ be the boundary of the relative convex hull of the cells. By construction, $\gamma'$ is a relatively convex tour enclosing $\gamma$, implying that any point seen by $\gamma$ is also seen by $\gamma'$. Furthermore, since $s\in \partial P$, it follows that $s$ cannot be in the interior of $P_{\gamma'}$, a subpolygon of $P$, thus $s\in \gamma'$.

We claim that $|\gamma'|$ is at most $|\gamma|+(8+4\sqrt{2})\delta n$. For each edge $e'$ of $\gamma'$ going from $\sigma_i$ to $\sigma_j$, we can bound its length by the sum of the length of the edge $e$ of $\gamma$ going from $\sigma_i$ to $\sigma_j$ ($\gamma'$ visits the cells containing the vertices of $\gamma$ in the same order) and at most two connections from endpoints of $e$ to vertices of $\sigma_i, \sigma_j$, which is no more than $2\sqrt{2}\delta$, see Figure~\ref{fig:snapping_bounding_length}, right. Additionally, the part of $\gamma'$ along the perimeters of $\sigma_1,\sigma_2,\ldots$ is no longer than $8\delta n$. Hence, $|\gamma'|\le |\gamma| + (8+4\sqrt{2})\delta n$.
\end{proof}

    \begin{figure}[h]
        \centering
        \includegraphics[width=1\textwidth]{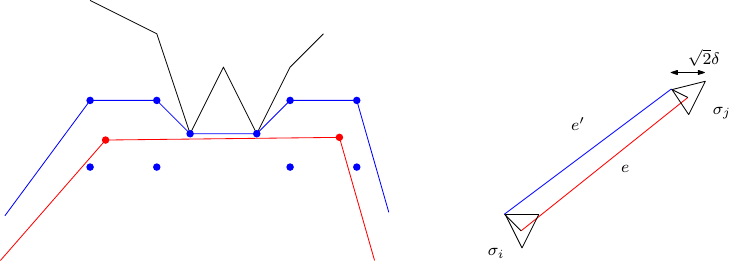}
        \caption{Left: $\gamma'$ (blue) is the relative convex hull of the vertices (blue) of the cells that contain convex vertices of $\gamma$ (red). Right: Each edge of $\gamma'$ that traverses between two different cells $\sigma_i, \sigma_j$ by triangle inequality, is no longer than the edge of $\gamma$ between the same cells plus at most two connections to two vertices of $\sigma_i, \sigma_j$.}
        \label{fig:snapping_bounding_length}
    \end{figure}

From Lemma \ref{lem:snapping}, if $\delta = O\left(\frac{\varepsilon |\gamma|}{n}\right)$, then for approximation purposes within factor $(1+\varepsilon)$, it suffices to search for a tour whose vertices come from $S_{\delta,L}$. In fact, our algorithm returns a tour no longer than $(1 + \varepsilon_1)|\gamma|$ for any $\varepsilon_1 > 0$; however, due to discretization of the area quota, we only guarantee the tour sees at least $(1-\varepsilon_2)A$ for any $\varepsilon_2 > 0$.

We now establish an ordering on the point set $S_{\delta,L}$, so that a relatively convex chain of the candidate points moves in increasing order. First, we compute $\mathcal{T}$, the tree of shortest paths rooted at $s$ to all the candidate points; this takes $O(|S_{\delta,L}|)$ time~\cite{guibas1986linear}. The path from $s$ to a candidate point $s'$ in $\mathcal{T}$ is the geodesic shortest path $\pi_P(s,s')$. Define a \textit{geodesic angular order} as follows: for two candidate points $s_i, s_j$, if $s_i$ is to the left of the extended geodesic shortest path between $s$ and $s_j$, i.e. $\pi_P(s, s_j)$ with the last segment extending up to $\partial P$, then $s_i$ precedes $s_j$. In case of ties, we break ties by increasing distance to $s$. For each reflex vertex $r_i$ of $P$, we add another candidate $s_{r_i}$ to the list to account for the possibility that $r_i$ can appear as two different vertices of a relatively convex polygonal chain; $s_{r_i}$ obeys the aforementioned geodesic angular order but precedes every candidate point in the subtree of $\mathcal{T}$ rooted at $r_i$. Sort the candidates accordingly, then append $s_m:= s_1$ to the end of the sorted list. Any relatively convex chain with vertices sequence oriented clockwise $(s, s_{i_1}, s_{i_2}, \ldots)$ has $1 < i_1 < i_2 < \ldots$. Without loss of generality, we consider any relatively convex polygonal chain to be oriented clockwise.

Next, we examine the optimal substructure of the problem.
\begin{lemma}[\cite{buchin2020geometric}]
\label{lem:visgeopath}
    The visibility region of the geodesic shortest path $\pi_P(s, s_j)$ is the inclusion-wise minimal set among all visibility regions of all paths from $s$ to $s_j$.
\end{lemma}
Let $C$ be a relatively convex polygonal chain from $s$ to a candidate point $s_j$, and let $s_i$ be the vertex of $C$ immediately preceding $s_j$. We identify the overlap of visibility between the segment $s_is_j$ and $C_{s_i}$, the subchain of $C$ from $s$ to $s_i$ in Lemma~\ref{lem:opt_chain_seg}.

\begin{lemma}
\label{lem:opt_chain_seg}
    $V(C_{s_i})\cap V(s_is_j) = V(\pi_P(s,s_i))\cap V(s_is_j)$.
\end{lemma}

\begin{proof}
    Refer to Figure~\ref{fig:optimal_substructure_chain_seg}. Let $x\in P$ be a point seen by both $\pi_P(s, s_i)$ and $s_is_j$. Since $x\in V(\pi_P(s, s_i))$, it follows that  $x\in V(C_{s_i})$ (Lemma \ref{lem:visgeopath}). Thus, $x\in V(C_{s_i})\cap V(s_is_j)$ and $ V(\pi_P(s, s_i))\cap V(s_is_j)\subseteq V(C_{s_i})\cap V(s_is_j)$.

    On the other hand, let $x\in P$ be seen by both $C_{s_i}$ and $s_is_j$. Since $x\in V(C_{s_i})\cap V(s_is_j)$ there exists $x_1\in C_{s_i}$ and $x_2 \in s_is_j$ such that $xx_1$ and $xx_2$ are contained within $P$. Thus, the (pseudo)triangle $xx_1x_2$ is contained within $P$ since $P$ has no holes. By our ordering scheme, $s_j$ is to the right of $\pi_P(s,s_i)$ with the last segment extended up to $\partial P$, while $C_{s_i}$ is to the left of it. This implies that in the relatively convex polygon $P_{C\cup\pi_P(s_, s_j)}$, $x_1, x_2$ are in opposite sides with respect to $\pi_P(s, s_i)$. As we pivot a line of sight around $x$ from $x_1$ to $x_2$, it must intersect $\pi_P(s,s_i)$ at some point due to continuity as well as (relative) convexity, therefore $\pi_P(s,s_i)$ sees $x$. Hence, $V(C_{s_i})\cap V(s_is_j) \subseteq V(\pi_P(s, s_i))\cap V(s_is_j)$.
\end{proof}

    \begin{figure}[h]
        \centering
        \includegraphics[width=1.05\textwidth]{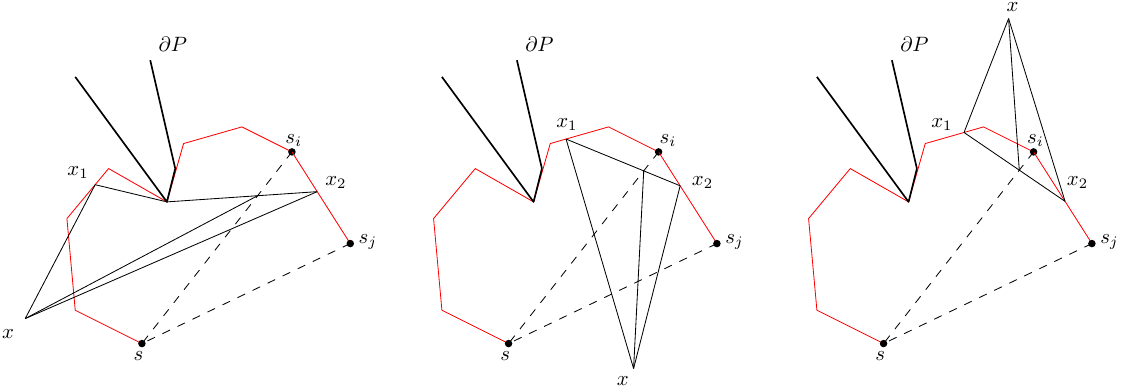}
        \caption{Proof of Lemma \ref{lem:opt_chain_seg}.}
        \label{fig:optimal_substructure_chain_seg}
    \end{figure}

Based on Lemma~\ref{lem:opt_chain_seg}, the overlap of visibility between the segment $s_is_j$, for $i < j$, and any relatively convex chain $C_{s_i}$ from $s$ to $s_i$ does not depend on the vertices between $s$ and $s_i$. This leads to the optimal substructure utilized by our dynamic programming algorithm. 

\begin{lemma}
\label{lem:opt_struct_quota}
    $C$ is a shortest relatively convex polygonal chain from $s$ to $s_j$ that sees at least some area $\overline{A}$ if and only if $C_{s_i}$ is a shortest relatively convex polygonal chain from $s$ to $s_i$ that sees at least area $\overline{A} - |V(s_is_j)\setminus V(\pi_P(s, s_i))|$.
\end{lemma}

\begin{proof}
    We write $V(C)$ as the union of 2 disjoint sets $V(C_{s_i})$ and $V(s_is_j)\setminus V(C_{s_i})$. Notice that
    \begin{align*}
        &V(s_is_j)\setminus V(C_{s_i}) =V(s_is_j)\setminus (V(C_{s_i})\cap V(s_is_j)) = V(s_is_j)\setminus (V(\pi_P(s, s_i))\cap V(s_is_j)) \\& = V(s_is_j)\setminus V(\pi_P(s, s_i)),
    \end{align*}
    therefore $|V(C_{s_i})| \ge \overline{A} - |V(s_is_j)\setminus V(\pi_P(s, s_i))|.$
    As a result, $C_{s_i}$ must be the shortest chain to achieve a visibility area of $\overline{A} - |V(s_is_j)\setminus V(\pi_P(s, s_i))|$, since the existence of a shorter chain contradicts the optimality of $C$, and vice versa.
\end{proof}

A subproblem in the dynamic program is determined by a candidate point $s_j$ and an area quota $\overline{A}$. Let $\pi(s_j, \overline{A})$ denote the length of a shortest relatively convex polygonal chain from $s$ to $s_j$ that can see area at least $\overline{A}$; and let $C(s_j, \overline{A})$ denote the associated optimal chain. Initialize $\pi(s, |V(s)|) = 0$. The Bellman recursion for each subproblem with $j = 1, 2, \ldots, m$ and all values of $\overline{A}$ would be given as follows, for all $\overline{i} < j$ such that $s_j$ sees $s_{\overline{i}}$ and $C(s_{\overline{i}}, \overline{A} - |V(s_{\overline{i}}s_j)\setminus V(\pi_P(s, s_{\overline{i}}))|)\cup s_{\overline{i}}s_j$ is relatively convex:
    \begin{align*}
    \label{eq:recursion_mainDP}
    i &= \underset{\overline{i}}{\argmin}\left\{\pi(s_{\overline{i}}, \overline{A} - |V(s_{\overline{i}}s_j)\setminus V(\pi_P(s, s_{\overline{i}}))|) + |s_{\overline{i}}s_j|\right\},\\
        \pi(s_j, \overline{A}) &= \pi(s_i, \overline{A} - |V(s_is_j)\setminus V(\pi_P(s, s_i))|) + |s_is_j|,\\
        C(s_j, \overline{A}) &= C(s_i, \overline{A} - |V(s_is_j)\setminus V(\pi_P(s, s_i))|)\cup s_is_j.
    \end{align*}
    Finally, return $C(s_m, A)$. Correctness of the algorithm follows from the principle of optimality.

    \begin{figure}[h]
        \centering
        \includegraphics[width=0.5\textwidth]{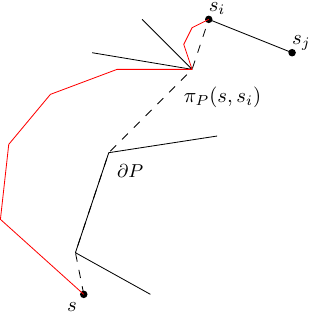}
        \caption{Solving subproblem $(s_j, \overline{A})$.}
        \centering\label{fig:dp_recursion}
    \end{figure}

Note that there always exists an optimal solution $i$ to the above recursion such that $C(s_i,~\overline{A}~-~|V(s_is_j)\setminus~V(\pi_P(s, s_i))|)\cup~s_is_j$ is relatively convex. Otherwise, we can shortcut $C(s_i,~\overline{A}~-~|V(s_is_j)\setminus~V(\pi_P(s, s_i))|)\cup~s_is_j$ by connecting $s_j$ to the closest reflex vertex (of $P$) or the point of tangency in $C(s_i, \overline{A} - |V(s_is_j)\setminus V(\pi_P(s, s_i))|)$.

% xxx rephrase some?    
Since $\overline{A}$ can take values from a continuous set, it is impractical to tabulate all such values. Instead, we bucket $A$ into uniform intervals, and let the subproblems be defined by interval endpoints. We round down the area of any visibility polygon to the nearest interval. Since we sum up the area of at most $2(n-3) + \frac{2L}{\delta}$ visibility polygons (each of the $n-3$ diagonals in the triangulation of $P$ and $\frac{L}{\delta}$ horizontal/vertical grid lines potentially has at most 2 vertices of the tour returned by the dynamic programming algorithm since we enforce relative convexity), if we denote by $I$ the length of each interval, the area lost by rounding down is at most $I\left(2(n-3) + \frac{2L}{\delta}\right)$. We run the algorithm on the ``rounded down'' instance with area quota $A~-~I\left(2(n-3) + \frac{2L}{\delta}\right)$, and since the optimal solution of the original instance is a feasible solution of the new instance, the algorithm returns a tour no longer than an optimal tour $\gamma'$ (that sees at least area $A$).

It remains to compute an appropriate $L$ such that an optimal tour $\gamma$ is contained within the bounding box of the grid.  Denote by $C_g(r)$ the geodesic disk of radius $r$ centered at $s$ (the locus of points whose length of the geodesic path to $s$ is no greater than $r$). Let $r:=r_{\min}$ where $r_{\min}$ is the smallest value of $r$ such that $|V(C_g(r))| = A$; then, $r$ is a lower bound on $|\gamma|$, since a tour of length $r$ has geodesic radius at most $r/2$ and thus cannot see an area of $A$. We can compute $r$ by the ``visibility wave'' methods in \cite{quickest} by considering all $O(n^2)$ edges of the visibility graph $G_v$ of $P$ (nodes are vertices of $P$ and two nodes are adjacent if they are visible to one another); we have a sequence of visibility edges hit by $C_g(r)$ for the first time in the process of increasing $r$, obtained by sorting the distance from every visibility edge to $s$ in $O(n^2\log n)$ time.

Moreover, $|\partial C_g(r)| + 2r$ is an upper bound on $|\gamma|$, since if the watchman goes from $s$ to $\partial C_g(r)$ ($s$ may not be on $\partial C_g(r)$), follows along $\partial C_g(r)$ then goes back to $s$, he sees an area of $A$. We argue that $|\partial C_g(r)| + 2r = O(nr)$ as follows: $\partial C_g(r)$ consists of polygonal chains that are portions of $\partial P$ and at most $n$ circular arcs; the circular arcs have total length at most $2n\pi r$. For each segment in the polygonal part of $\partial C_g(r)$, we can bound its length by the sum of geodesic distances from its endpoints to $s$ (triangle inequality), which is no more than $2r$. There are at most $n$ segments in the polygonal portions of $\partial C_g(r)$, therefore their total length is no longer than $2nr$, implying $|\gamma| \le |\partial C_g(r)| + 2r = 2nr + 2n\pi r + 2r \le 9nr$. 

We initialize $L:= r$ and run the dynamic program with the following $\delta$ and $I$:
\begin{align*}
    \delta = \frac{\varepsilon_1 L}{16 + 8\sqrt{2}n},\qquad
    I = \frac{\varepsilon_1\varepsilon_2A}{2(n-3)\varepsilon_1 + (32 + 16\sqrt{2})n}.
\end{align*}
Then, set $L := 2L$, and repeat until $L\ge 9nr$. At some point, we must have $|\gamma| \le L \le 2|\gamma|$, which means the approximate tour $\gamma'$ returned by the dynamic program will be such that $V|\gamma'|\ge (1-\varepsilon_2)|V(\gamma)|$ and $|\gamma'|\le (1+\varepsilon_1)|\gamma|$. We return the shortest tour out of all tours that achieve the visibility area quota as we increase $r$. Since $r\le |\gamma| = O(nr)$, the number of iterations of the doubling search is $O(\log n)$.
The resulting theorem is as follows:

\old{SoCG reviewer 2:
What is also curious is that Theorem 12 is strictly stronger than
Theorem 7, and there is no explanation why the manuscript first
presents (in nearly full detail) the weaker Theorem 7 instead of
concentrating on the results of Section 3 that lead to the stronger
results (and solve both BWRP and QWRP in one go).}

% Theorem 7:
\begin{theorem}
\label{thm:running_time_qwrp_dual}
    Given a starting point $s$, a dual approximation $\gamma'$ to an optimal solution $\gamma$ of the QWRP, with area quota $A$, in a simple polygon with $n$ vertices, satisfying $|V(\gamma')| \ge (1-\varepsilon_2)A$ and $|\gamma'| \le (1 + \varepsilon_1)|\gamma|$ for any $\varepsilon_1, \varepsilon_2 > 0$, can be computed in $O\left(\frac{n^5}{\varepsilon_1^5\varepsilon_2}\log n\right)$ time if $s\in\partial P$ or $O\left(\frac{n^9}{\varepsilon_1^9\varepsilon_2}\log n\right)$ time if $s\notin\partial P$\old{polynomial time (in $n, \frac1{\varepsilon_1}$ and $\frac1{\varepsilon_2}$)}.
\end{theorem}
\begin{proof}
    We triangulate $P$ in $O(n)$ time, then compute all intersections between polygon edges, diagonals of the triangulation and horizontal/vertical segments of the square grid. We bound the number of candidate turn points by $|S_{\delta,L}| = O\left(n +\frac{n}{\varepsilon_1}\cdot\frac{n}{\varepsilon_1} + \frac{n}{\varepsilon_1}\cdot n + \frac{n}{\varepsilon_1}\cdot n\right) = O\left(\frac{n^2}{\varepsilon_1^2}\right)$. After sorting the candidates by geodesic angular order, we compute and store, for each candidate $s_i$, $V(\pi_P(s, s_i))$. This can be done in $O(n)$ per candidate point recursively. Start by computing and storing $V(ss_i)$ for every child $s_i$ of $s$ in $\mathcal{T}$. Then, for each internal node $s_j$, compute the visibility polygon of the segment between $s_j$ and its parent node $p(s_j)$ in $O(n)$ time, and merge it with $V(\pi_P(s, p(s_j))$ (which is a polygon with at most $n$ vertices) in $O(n)$ time. This allows us to compute $V(s_is_j)\setminus V(\pi_P(s,s_i))$, for all pairs of candidates $s_i, s_j$ that see each other, in $O\left(\frac{n^5}{\varepsilon_1^4}\right)$ time, which clearly dominates the running time of computing and sorting $S_{\delta,L}$. Computing a starting lower bound to $|\gamma|$ for the doubling search using the visibility wave method takes $O(n^2\log n)$ time.

    For each iteration of the doubling search, we run the dynamic programming algorithm in which a subproblem is defined by a candidate point and a quota value, hence there are $O\left(\frac{n^2}{\varepsilon_1^2}\right)\cdot O\left(\frac{n}{\varepsilon_1\varepsilon_2}\right) = O\left(\frac{n^3}{\varepsilon_1^3\varepsilon_2}\right)$ subproblems. Each subproblem can be solved by recursing through $O\left(\frac{n^2}{\varepsilon_1^2}\right)$ previously solved subproblems and computing the area of visibility according to the Bellman recursion which takes $O(1)$ for each since we have all $V(s_is_j)\setminus V(\pi_P(s,s_i))$ pre-computed. Thus the overall running time is $O\left(\frac{n^5}{\varepsilon_1^5\varepsilon_2}\log n\right)$.
    
    When $s$ is interior to $P$, recall that an optimal tour $\gamma = (s,w_1,w_2,\ldots,w_k,s)$ can be made relatively convex by replacing $sw_1$ and $w_ks$ with $\pi_P(w_1, w_k)$. Thus we add a factor of $O\left(\frac{n^4}{\varepsilon_1^4}\right)$ to the running time by iterating over all pairs $(w_1, w_k)$, and finding a relatively convex chain $C$ from $w_1$ to $w_k$ minimizing $|C| + |sw_1| + |sw_2|$ while $|V(C)\cup V(sw_1)\cup V(sw_2)| \ge A$.
\end{proof}

\section{The BWRP in a Simple Polygon}
% In Appendix \ref{appendix:hardness} we prove the hardness of BWRP in simple polygons:
\label{sec:bwrp}
\begin{theorem}
\label{thm:weakly-np-hard-bwrp}
The BWRP in a simple polygon is weakly NP-hard.
    %\joe{do we want to say "weakly"? I added, for now}
\end{theorem}
\begin{proof}
    Consider an instance of \textsc{Knapsack}: Given a budget $B$, items of integral weights $(l_1,\ldots,l_m)$, and values $(V_1,\ldots,V_m)$, we seek a combination of items having maximum total value such that their collective weight is at most~$B$.

     We employ the same construction in the proof of Theorem~\ref{thm:weakly-np-hard-qwrp} as shown in Figure~\ref{fig:floating-np-hard-gadget}. We give the watchman a budget of $ml_0 + B$. If we can compute a route of length $ml_0 + B$ with the maximum visibility area in polynomial time, then we can compute a combination of items of total weight no greater than $B$ maximizing the total value in polynomial time. Thus, the BWRP is NP-hard.
\end{proof}
\subsection{Approximation algorithm for anchored BWRP}
For a given budget length $B>0$, and any fixed $\varepsilon>0$, we compute a route of length at most $(1 + \varepsilon)B$ that sees at least as much area as an area-maximizing route $\gamma$ of length $B$.
\begin{claim}
    Without loss of generality, we can assume that $B$ is less than the length of an optimal watchman route for $P$; otherwise, the solution is simply an optimal WRP tour. Hence, an optimal budgeted watchman route $\gamma$ is necessarily the shortest watchman route of the subpolygon $V(\gamma)$, and so $\gamma$ is relatively convex (otherwise, we can shortcut $\gamma$ and expand the remaining length budget to see more area, contradicting the optimality of $\gamma$).
\end{claim}

Since it suffices to consider only relatively convex routes for the BWRP, the observation in Lemma \ref{lem:opt_chain_seg} allows us to prove the structure of an optimal BWRP tour.

\begin{lemma}
\label{lem:opt_struct_budget}
    $C$ is a relatively convex polygonal chain from $s$ to $s_j$ of length at most $\overline{B}$ that sees the largest area possible if and only if $C_{s_i}$ is a relatively convex polygonal chain from $s$ to $s_i$ of length at most $\overline{B} - |s_is_j|$ that sees the largest area possible.
\end{lemma}
\begin{proof}
    Recall that in Lemma~\ref{lem:opt_struct_quota}, we proved that $|V(C)| = |V(C_{s_i})| + |V(s_is_j)\setminus V(\pi_P(s, s_i))|$. 
    This implies $|V(C_{s_i})|$ must be the maximum among all chains of length at most $\overline{B} - |s_is_j|$, since the existence of a chain that achieves more area of visibility contradicts the optimality of $|V(C)|$, and vice versa.
\end{proof}

We decompose $P$ by overlaying a triangulation and regular square grid of $\delta$-sized pixels within an axis-aligned square of size $B$-by-$B$, centered on $s$, then sort the set of candidates $S_{\delta,B}$ according to the geodesic angular order defined for the QWRP. Similarly, there exists a route $\gamma'$ of length at most $|\gamma| + (8 + 4\sqrt{2})\delta n$ with vertices in $S_{\delta,B}$.

Let $a(s_j, \overline{B})$ be the optimal area that a relatively convex polygonal chain from $s$ to $s_j$ that is no longer than $\overline{B}$ can see; and let $C(s_j, \overline{B})$ be the associated optimal chain of the subproblem $(s_j, \overline{B})$. Initialize $a(s, 0) = |V(s)|$. The Bellman recursion for each subproblem with $j = 1, 2, \ldots, m$ and all values of $\overline{B}$ would be given as follows, for all $\overline{i} < j$ such that $s_j$ sees $s_{\overline{i}}$ and $C(s_{\overline{i}}, \overline{B} - |s_{\overline{i}}s_j|)\cup s_{\overline{i}}s_j$ is relatively convex
    \begin{align*}
    \begin{split}
    i&=\underset{\overline{i}}{\argmax}\left\{a(s_{\overline{i}}, \overline{B} - |s_{\overline{i}}s_j|) + |V(s_{\overline{i}}s_j) \setminus V(\pi_P(s, s_{\overline{i}}))|\right\}\\
        a(s_j, \overline{B}) &= a(s_i, \overline{B} - |s_is_j|) + |V(s_is_j) \setminus V(\pi_P(s, s_i))|,\\
        C(s_j, \overline{B}) &= C(s_i, \overline{B} - |s_is_j|) \cup s_is_j.
        \end{split}
    \end{align*}
    Then, return $\gamma' := C(s_m, B + (8+4\sqrt{2})\delta n)$. 

    To bound the number of subproblems, we consider a partition of an interval of length $B + (8+4\sqrt{2})\delta n$ into uniform intervals, and round up the length of any segment $s_is_j$ to the nearest interval endpoint. Let $I$ be the length of each interval, we run the algorithm on a new instance with budget $B + (8+4\sqrt{2})\delta n + \left(2(n-3) + \frac{2B}{\delta}\right)I$ and the subproblems defined by intervals' endpoints. The optimal solution of the original instance is a feasible solution of the new instance; thus, we find a route seeing as much as the optimal route of the original instance. The values of $\delta$ and $I$ can be set as follows:
\begin{align*}
    \delta = \frac{\varepsilon B}{(16 + 8\sqrt{2})n},\qquad
    I = \frac{\varepsilon^2 B}{4(n-3)\varepsilon + (64 + 32\sqrt{2})n},
\end{align*}
so that $B + (8+4\sqrt{2})\delta n + \left(2(n-3) + \frac{2B}{\delta}\right)I \le (1+\varepsilon)B$.

\begin{theorem}
\label{thm:running_time_bwrp}
    Given a starting point $s$, a route of length at most $(1 + \varepsilon)B$ seeing at least as much area as is seen by an optimal route of length $B$ for the BWRP in a simple polygon  with $n$ vertices can be computed in $O\left(\frac{n^5}{\varepsilon^6}\right)$ time if $s\in\partial P$ or $O\left(\frac{n^9}{\varepsilon^{10}}\right)$ time if $s\notin\partial P$.
\end{theorem}
\begin{proof}
    The analysis of the running time of the dynamic program for the BWRP can be done almost identically to that of the dual approximation algorithm for the QWRP. A subproblem is defined by a candidate point and a budget value, resulting in $O\left(\frac{n^3}{\varepsilon^4}\right)$ subproblems in total, each of which can be solved in $O\left(\frac{n^2}{\varepsilon^2}\right)$ time. The doubling search is not required, since the budget value gives the size of the bounding box of the route. We obtain a running time of $O\left(\frac{n^5}{\varepsilon^6}\right)$ for the dynamic program for the BWRP if a starting point $s$ is given on the boundary, and $O\left(\frac{n^9}{\varepsilon^{10}}\right)$ if $s$ is interior to $P$.
\end{proof}
\subsection{From anchored BWRP to an FPTAS for anchored QWRP}
We can adapt the algorithm for the anchored BWRP above to obtain an FPTAS for the anchored QWRP. Let $\gamma$ be an optimal QWRP tour, suppose we have some $L$ such that $|\gamma| \le L \le 2|\gamma|$. We divide $L$ into $\displaystyle\ceil{\frac{2}{\varepsilon}}$ uniform intervals, each of length no greater than $\varepsilon|\gamma|$. The smallest interval endpoint $L'$ that is no smaller than $|\gamma|$, is a $(1 + \varepsilon)$-approximation to $|\gamma|$. We can iterate through interval endpoints as the budget constraint and use the approximation algorithm for the BWRP to compute a route of length at most $(1 + \varepsilon)L'$ that sees as much as an area-maximizing route of length $L'$ does, which sees more area than does~$\gamma$. The result is the following theorem, which, in contrast with the earlier Theorem~\ref{thm:running_time_qwrp_dual}, is not a dual approximation (allowing for a relaxation of the quota constraint), but an FPTAS for optimizing the length, subject to a hard constraint on the area seen. The running time of the algorithm in the dual approximation of Theorem~\ref{thm:running_time_qwrp_dual}, however, is better than that of the FPTAS, so it may be preferred in some settings.
% xxx say more explicitly the running time comparison??

\old{SoCG reviewer 2:
What is also curious is that Theorem 12 is strictly stronger than
Theorem 7, and there is no explanation why the manuscript first
presents (in nearly full detail) the weaker Theorem 7 instead of
concentrating on the results of Section 3 that lead to the stronger
results (and solve both BWRP and QWRP in one go).}

% Theorem 12:
\begin{theorem}
\label{thm:fptas_qwrp} Given a starting point $s$, an approximation $\gamma'$ to an optimal solution $\gamma$ of the QWRP, with area quota $A$, in a simple polygon with $n$ vertices, satisfying $|V(\gamma')| \ge A$ and $|\gamma'| \le (1 + \varepsilon)|\gamma|$ for any $\varepsilon > 0$, can be computed in $O\left(\frac{n^5}{\varepsilon^6}\log n\log\frac{1}{\varepsilon}\right)$ if $s\in\partial P$ or $O\left(\frac{n^9}{\varepsilon^{10}}\log n\log\frac{1}{\varepsilon}\right)$ if $s\notin\partial P$.
\end{theorem}
\begin{proof}
    We first obtain a lower bound on $|\gamma|$ by computing the minimum radius $r_{\min}$ of a geodesic ball seeing at least the area quota $A$; let $L:=r_{\min}$. Perform a binary search on the values $\displaystyle 0, \frac{L}{\ceil{\frac{2}{\varepsilon}}}, \frac{2L}{\ceil{\frac{2}{\varepsilon}}}, \ldots, L$, as choices for the input budget to the algorithm for the anchored BWRP with a budget stretch of $(1 + \varepsilon)$. In $O\left(\frac{n^5}{\varepsilon^6}\log\frac{1}{\varepsilon}\right)$ time if the starting point $s$ is on $\partial P$ or $O\left(\frac{n^9}{\varepsilon^{10}}\log\frac{1}{\varepsilon}\right)$ if $s$ is interior to $P$, we can find a budget value for which the algorithm returns a tour with visibility area greater than the quota $A$, if there is one. If no such value exists, set $L:=2L$ and repeat the binary search. Ultimately, we find a $(1 + \varepsilon)^2$-approximation to the QWRP. Let $2\varepsilon + \varepsilon^2 = \varepsilon'$, then $\frac{1}{\varepsilon}= \Theta\left(\frac{1}{\varepsilon'}\right)$ as $\varepsilon$ and $\varepsilon'$ approach 0. Thus in total running time of $O\left(\frac{n^5}{\varepsilon'^6}\log n\log\frac{1}{\varepsilon'}\right)$ or $O\left(\frac{n^9}{\varepsilon'^{10}}\log n\log\frac{1}{\varepsilon'}\right)$ depending on the location of $s$, we can compute a $(1 + \varepsilon')$-approximate solution to the QWRP.
\end{proof}

\section{Floating QWRP and BWRP}
\label{sect:floatin}
When the starting point $s$ of the tour is not specified (the so called ``floating'' case), the WRP tends to be trickier: known algorithms for the floating WRP are $O(n)$-factor slower than in the non-floating case \cite{first,carlsson1999finding,chin1991shortest,corrigendum,dror2003touring,ntafos1994optimum,tan2001fast}. If the optimal tour is not convex (but only relatively convex), one can iterate through all reflex vertices of $P$ as choices for $s$, and thus reduce the floating version to the basic WRP; the same can be done for QWRP and BWRP. Thus, the remaining challenge is to find the shortest (strictly, not just relatively) convex tour. 

\begin{figure}[h!]\centering\includegraphics[page=4]{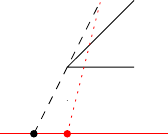}\caption{What is seen from the interior of a segment does not change as the segment is moved locally: whatever was hidden, but becomes seen from an interior point (red), was seen by a neighboring point (black) before the move.}\label{fig:floating}\end{figure}

Any convex polygon can be outer-approximated by a convex polygon with a constant number of vertices: Dudley's approximation \cite{bronshteyn1975approximation,dudley1974metric,har2019proof} implies that for any length-$L$ tour $\gamma$, there is a length-$(1+\varepsilon)L$ tour $\gamma'$ with $O\left(\frac{1}{\sqrt\varepsilon}\right)$ vertices that sees at least as much area as $\gamma$ does. %https://www.maths.tcd.ie/~hmigca-18/slides/mount-2.pdf
To find $\gamma'$ (either for QWRP or BWRP), we use the techniques from \cite{ntafos1994optimum}: For each of the $O(n^4)$ cells of the visibility decomposition $D$ (the visibility graph $G_v$ of $P$ as defined Section~\ref{sec:dualapproxqwr}, with maximally extended edges), points within the cell have visibility polygons that are combinatorially equivalent, implying that the area seen by any point in the cell is given by the same formula. Moreover, if each vertex of a tour sits within a fixed cell of $D$ and each edge of the tour passes through the same set of cells, the total area seen from the tour is given by the same formula: the interiors of the edges do not add to the area seen, so the total seen area is a function, $f(v_1,\dots,v_k)$, of only the positions $v_i$ of the tour's $k=O\left(\frac{1}{\sqrt\varepsilon}\right)$ vertices (Fig.~\ref{fig:floating}). 
We further decompose the cells of $D$ by lines through all of $D$'s vertices. If the vertices of the tour are in the same cells of this refined $O(n^8)$-complex decomposition $D'$, then the edges of the tour pass through the same cells of $D$. We iterate through all $O(n^{8k})$ placements of vertices of the tour into cells of $D'$. For each placement, finding $\vec v$ maximizing the seen area, $f(\vec v)$, amounts to solving an $O(k)$-sized system of polynomial equations having $O(k)$ algebraic degree (the Lagrangian of the problem will contain the constraint that the tour's length is $L$, consisting of $k$ terms and will have to be squared $O(k)$ times before becoming a polynomial). The solution can be found in $k^{O(k)}$ time \cite[Section~3.4]{canny1988complexity}.
We summarize in the following theorems:

% be more explicit about running times?? xxx
% xxx Linh: I put the running times in the theorem

\begin{theorem}
Let $\gamma$ be an optimal QWRP (no starting point) tour. In $n^{O\left(\frac{1}{\sqrt\varepsilon}\right)}$ time, a tour of length at most $(1+\varepsilon)|\gamma|$ can be found that sees at least as much area as does $\gamma$.
\end{theorem}
\begin{theorem}
In $n^{O\left(\frac{1}{\sqrt\varepsilon}\right)}$ time, a BWRP (no starting point) tour of length $(1+\varepsilon)L$ can be found that sees at least as much area as does any tour of length $L$.
\end{theorem}
% Refer to Appendix \ref{appendix:runtime} for the proof of Theorem \ref{thm:fptas_qwrp}.

\section{Domains that are a Union of Lines }

We consider the QWRP and the BWRP in a domain $P$ that is a connected union (arrangement) of lines; it suffices to truncate the lines within a bounding box that encloses all vertices of the line arrangement, so that $P$ is bounded. Such domains, which are a special case of polygons with holes, have been studied in the context of the watchman route problem~\cite{dumitrescu2014minimum, dumitrescu2014watchman}. In this setting, the QWRP seeks to minimize the length of a route contained within $P$ that visits at least a specified number of (truncated) lines, and the BWRP seeks to maximize the number of lines visited by a route within $P$ of length at most some budget. In this section we do not assume a depot $s$ is specified.

For a set of lines $\mathcal{L}$, let $\mathcal{A}(\mathcal{L})$ denote the arrangement formed by $\mathcal{L}$. We assume that not all of the lines are parallel so that the union is connected. All of a line can be seen from any point incident on it. Let $G(\mathcal{L})$ denote the weighted planar graph with vertex $V(\mathcal{A}(\mathcal{L}))$ of intersections between lines. 
Two vertices in the graph are connected by an edge with Euclidean weight if they share the same edge in $\mathcal{L}$.
Since the watchman is constrained to travel within $\mathcal{A}(\mathcal{L})$ and the only time the route can have turning points not in $V(\mathcal{A}(\mathcal{L}))$ is when it traces out the same edge of $G(\mathcal{L})$ consecutively in opposite directions, turning somewhere interior to that edge. However, then we can shortcut that portion of the route altogether while maintaining visibility coverage, it is easy to see that any optimal route for BWRP or QWRP is polygonal and its vertices is a subset of $V(\mathcal{A}(\mathcal{L}))$.

We first explain our results for the QWRP, then we use them to solve the BWRP.
The following observation is essential to our algorithm: a line intersects a tour $\gamma$ if and only if it intersects the convex hull of $\gamma$. We define the \emph{quota intersecting convex hull} problem as follows: compute a cyclic sequence $(v_1, v_2, \ldots, v_h)$ of vertices $v_i \in V(\mathcal{A}(\mathcal{L}))$ in convex position such that the number of lines intersecting the convex polygon $(v_1, v_2, \ldots, v_h)$ is at least some specified $Q > 0$ and $\sum_{i}^h|\pi(v_i, v_{i+1})|$ is minimized $(v_{h+1} = v_1)$, where $\pi(s,t) = \pi_G(s,t)$ is the shortest path connecting $s$ and $t$ in $G(\mathcal{L})$. We show the relationship between the quota intersecting convex hull problem and the QWRP in an arrangement of lines.

\begin{lemma}
    An optimal solution to the quota intersecting convex hull problem yields an optimal solution to the QWRP in an arrangement of lines.
\end{lemma}

\begin{proof}
    Suppose $(v_1, v_2, \ldots, v_h)$ is an optimal solution to the quota intersecting convex hull problem of length $L$ intersecting $Q$ lines. We concatenate $\pi(v_1, v_2), \ldots,  \pi(v_{h-1}, v_{h})$ and $\pi(v_h, v_1)$ to form $\gamma$. Every line intersecting the convex polygon $(v_1, v_2, \ldots, v_h)$ must intersect $\gamma$ as well. Thus, $\gamma$ is a route of length $\sum_{i}^h|\pi(v_i, v_{i+1})| = L$ seeing $Q$ lines.

    We claim that there is no solution $\gamma'$ to the QWRP intersecting $Q$ lines that is strictly shorter than $\gamma$. Suppose to the contrary, take the convex hull of $\gamma'$, which has vertices in $V(\mathcal{A}(\mathcal{L}))$ since vertices of $\gamma'$ are in $V(\mathcal{A}(\mathcal{L}))$. The vertices of the convex hull of $\gamma'$ form a cyclic sequence that is feasible for the quota intersecting convex hull problem, and the length is exactly $|\gamma'|$ (or $\gamma'$ could be shortened while still intersecting $Q$ lines), which is strictly smaller than $L$. Thus, $\gamma'$ yields a feasible cyclic sequence intersecting $Q$ lines while the length is shorter than $\gamma$, violating the assumption that $(v_1, v_2, \ldots, v_h)$ is optimal.
\end{proof}
Note that there can be many optimal solutions to the quota intersecting convex hull problem yielding the same tour (Figure~\ref{fig:maximum_intersecting_convex_hull}).

    \begin{figure}[h]
        \centering
        \includegraphics[width=0.35\textwidth]{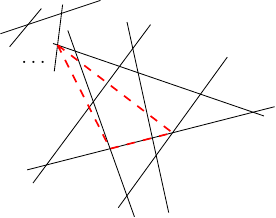}\hspace{2cm}
        \includegraphics[width=0.35\textwidth]{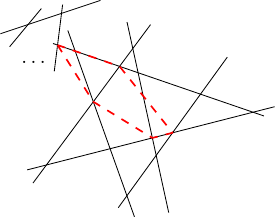}\\
        \vspace{1cm}
        \includegraphics[width=0.35\textwidth]{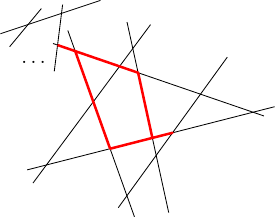}
        \caption{Multiple optimal solutions to the quota intersecting convex hull problem corresponding to the same optimal solution of the QWRP.}
        \label{fig:maximum_intersecting_convex_hull}
    \end{figure}

We give a dynamic programming algorithm to solve the quota intersecting convex hull problem. Fix one vertex to be the lowest vertex, let that vertex be $v_1$. We will examine all possible choices of $v_1$, and find the optimal cyclic sequence with each choice. Let $\{v_2, v_3, \ldots, v_{m-1}\}$ be the list of vertices above $v_1$ sorted by increasing angle with the left horizontal ray passing through $v_1$, breaking ties by increasing distance to $v_1$. Then, set a new element $v_m := v_1$ and append it to the list. Thus, an optimal cyclic sequence $(v_1, v_{i_2}, v_{i_3}, \ldots, v_m)$ has $1 < i_2 < i_3 < \ldots < m$. We can restrict ourselves to ordered pairs $(v_i, v_j)$ of consecutive vertices where $1\le i < j \le m$ and either $i\ne 1$ or $j\ne m$ in the sequence.

For $1\le i < j \le m$ and either $i\ne 1$ or $j\ne m$, denote by $\mathcal{L}_{i,j}$ the lines that intersect the segment $v_iv_j$ (including at the endpoints $v_i, v_j$). Each vertex $v_j$ and a quota value $\overline{Q}$ define a subproblem.  Let $\pi(v_j, \overline{Q})$ be the shortest length of a sequence of vertices in convex position from $(v_1, \ldots, v_j)$ intersecting at least $\overline{Q}$ lines. Starting from $v_1$, we initialize $\pi(v_1,|\mathcal{L}_{1,1}|) = 0$ with the associated sequence $(v_1)$. For $j=2,\ldots,m$ and $\overline{Q} = 1, 2, \ldots, n$, we solve the subproblems $(v_j,\overline{Q})$ by the following Bellman recursion, for all $i < j$ such that the sequence associated with $(v_i, \overline{Q} - |\mathcal{L}_{i,j}\setminus \mathcal{L}_{1,i}|)$ and $v_j$ are in convex position (Figure~\ref{fig:subproblem_lines})
\begin{equation*}
\label{eq:recursion_dp_lines}
    \pi(v_j,\overline{Q}) = \min\limits_{i}\left\{\pi(v_i, \overline{Q} - |\mathcal{L}_{i,j}\setminus \mathcal{L}_{1,i}|) + |\pi(v_i, v_j)| \right\}.
\end{equation*}

    \begin{figure}[h]
        \centering
        \includegraphics[width=0.8\textwidth]{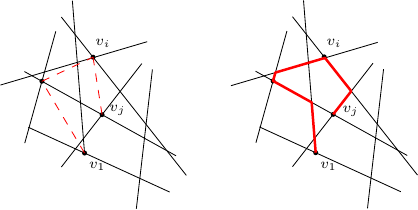}
        \caption{Solving subproblem $(v_j, \overline{Q})$. The sequence of vertices in convex position is drawn on the left with a dashed red chain, and the corresponding part of $\gamma$ is drawn with solid red segments on the right.}
        \label{fig:subproblem_lines}
    \end{figure}
Correctness of the algorithm follows from these two claims:
\begin{itemize}
    \item Given a sequence of vertices in convex position $(v_1, \ldots, v_i, v_j)$, any line intersecting both $\pi(v_i,v_j)$ and the subsequence from $v_1$ to $v_i$, $(v_1, \ldots, v_i)$ must intersect the segment $v_1v_i$ and vice versa due to continuity and convexity. If a sequence $(v_1, \ldots, v_i, v_j)$ is the shortest among all sequences from $v_1$ to $v_j$ intersecting at least $\overline{Q}$ lines, then the subsequence from $v_1$ to $v_{i}$ is the shortest sequence from $v_1$ to $v_i$ intersecting $\overline{Q} - |\mathcal{L}_{i,j}\setminus \mathcal{L}_{1,i}|$.

    \item The sequence $(v_1, \ldots, v_i)$ associated with optimal solution $i$ to the Bellman recursion is such that $(v_1, \ldots, v_i)\cup (v_j)$ are in convex position.
\end{itemize}

\begin{theorem}
\label{thm:lines}The QWRP and the BWRP in an arrangement of lines can be solved in $O(n^7)$ time.
\end{theorem}
\begin{proof}
    Since there are $n$ lines, there are $O(n^2)$ vertices in $V(\mathcal{A}(\mathcal{L}))$. We pre-compute $\pi(v_i, v_j)$ for all pairs $v_i, v_j$ in $O(n^6)$ time and store them for faster access during the algorithm. Also, pre-computing and storing $|\mathcal{L}_{i,j}\setminus\mathcal{L}_{k,i}|$ for all triples of vertices take $O(n^7)$ time.

To solve the quota intersecting convex hull problem, for each of the $O(n^2)$ choices of $v_1$, we sort the vertices above $v_1$ in $O(n^2\log n)$ time. There are $O\left(n^3\right)$ subproblems ($O(n^2)$ vertices above $v_1$ and $n$ values of the quota) once we fix $v_1$. A subproblem can be solved in $O(n^2)$ time, which includes considering all subproblems of lower indices and solving the Bellman recursion for each, which takes constant time, since we have all $|\mathcal{L}_{i,j}\setminus\mathcal{L}_{k,i}|$ and $\pi(v_i, v_j)$ pre-computed. Thus, in $O\left({n^7}\right)$ time, we can obtain a solution to the QWRP.

In terms of the BWRP (with given length budget $B$), we solve the QWRP with the quota $Q = n$. Then, simply perform a binary search over all solved subproblems $(v_m, \overline{Q})$ with $\overline{Q} = 1, 2, \ldots, n$, return the first value of $\overline{Q}$ such that $\pi(v_m, \overline{Q}) \le B, \pi(v_m, \overline{Q} + 1) > B$ and the associated cyclic sequence. Hence, the BWRP in an arrangement of lines can also be solved in $O(n^7)$ time.
\end{proof}
%xxx Made a remark about how this improves upon Adrian's SWAT paper
\begin{remark}
    When $Q = n$, the algorithm solves the classic WRP in an arrangement of lines, thus our result improves upon the $O(n^8)$ solution in \cite{dumitrescu2014watchman}. 
\end{remark}
% Appendix \ref{appendix:runtime} provides the proof of Theorem \ref{thm:lines} in details.
\section{The QWRP and BWRP in a Polygon With Holes}
\subsection{Hardness of approximation}
\begin{theorem}
\label{thm:hardness-approx-qwrp-holes}
    The QWRP in a polygon with holes cannot be approximated, in polynomial time, within a factor of $c\log n$ for some constant $c > 0$, unless P = NP.
\end{theorem}
\begin{proof}
    In~\cite{mitchell2013approximating}, Mitchell proved that the WRP admits no polynomial time approximation algorithm with a factor of $c\log n$ or better for some constant $c > 0$, unless P = NP. Since the WRP is a special case of the QWRP, in which the area quota is the area of whole polygonal domain, a similar lower bound of approximation is implied.
\end{proof}
    \begin{theorem}
    \label{thm:hardness-approx-bwrp-holes}
    The BWRP in a polygon with holes cannot be approximated, in polynomial time, within a factor of $\left(1 - \varepsilon\right)$ for arbitrary $\varepsilon > 0$, unless P = NP.
\end{theorem}
\begin{proof}
    Our proof is based on a reduction from the \textsc{Max $k$-Vertex Cover} problem in cubic graphs: Given a graph $G = (V,E)$ where every node has degree 3, and a positive integer $k$, find a subset of $k$ nodes such that the number of edges incident with at least a node in the subset is maximized. Petrank proved in \cite{petrank1994hardness} that unless P = NP, there is no polynomial time algorithm that approximates the \textsc{Max $k$-Vertex Cover} problem in cubic graphs to a factor of $1 - \varepsilon$ for arbitrarily small $\varepsilon > 0$.

    We construct an instance of the BWRP in a polygon with holes $P$ as depicted in Figure~\ref{fig:reduction_maxkvertexcover} (not to scale).
    %xxx changed some annotations on the figure h->1, 1->l, the reviewers didn't catch these errors
    \begin{center}
    \begin{figure}[h]
        \centering
        \includegraphics[width=0.95\textwidth]{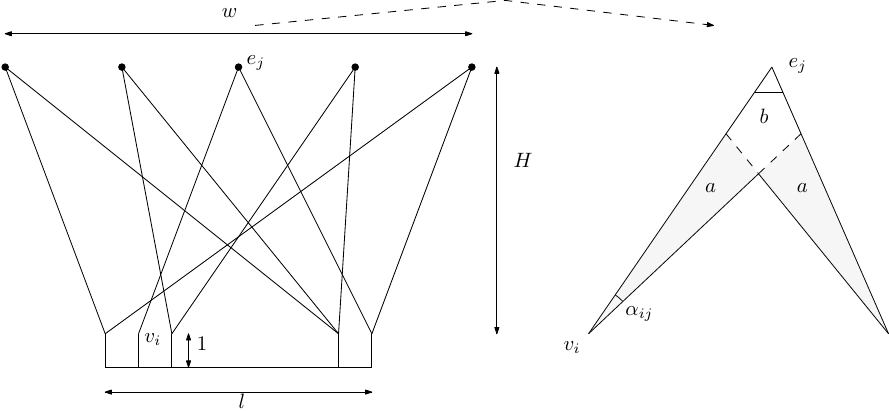}
        \caption{Reduction from the \textsc{Max $k$-Vertex Cover} problem in cubic graphs.}
       \label{fig:reduction_maxkvertexcover}
    \end{figure}
    \end{center}
    
    Black segments represent narrow corridors. From the bottom corridor of length $l \ll 1$, there are $n$ small vertical corridors each of height $1$; these are the node gadgets. Black circles, spaced evenly along a horizontal width $w$ at height $H$, represent edges. If an edge and a node are incident in $G$, we connect them accordingly with a triangular corridor. Each black circle is connected to exactly two node gadgets, and an edge gadget is the rhombus-shaped intersection of two triangular corridors. The bottom corridor, as well as the node gadgets, have infinitesimal width and negligible area. We align the edge gadgets so that they can only be seen from the top of a node gadget.

    For an edge gadget $e_j$ incident with node gadget $v_{i}$, let $\alpha_{ij}$ be the angle of the triangular corridor connecting $v_{i}$ and $e_j$ at (the top of) $v_{i}$. We adjust the width $w$ along which and the height $H$ at which the edge gadgets are placed as well as the angles $\alpha_{ij}$'s so that
    \begin{itemize}
        \item The triangular regions in node-edge corridors that are seen by only one node gadget incident with the edge gadget (shaded gray in Figure~\ref{fig:reduction_maxkvertexcover}) have constant area $a$.
        \item The smallest angle $\alpha_{ij}$ is no smaller than a constant (small) angle $\overline{\alpha}$.
        \item While corridors connecting different edge gadgets to different node gadgets can cross (we place the edge gadgets so that no three corridors intersect at a common area; thus, an intersection is only seen twice), the lowest crossing is at a height greater than the budget~$B$.  Additionally, the highest crossing must happen far away from the edge gadgets, so that total area of all crossings ($O(k^2)$ of them) is minuscule compared to the area of one edge gadget, in particular the total area of all crossing should be no greater than $\varepsilon b$.
    \end{itemize}
    The requirements above can be satisfied by making $H$ and $w$ arbitrarily large.
        
    Now, each edge gadget may have a different area, let $b$ be the minimum area of any edge gadget. Note that by trigonometry $b = \Theta(\sin(\alpha_{ij})a) = \Theta(\sin(\overline{\alpha})a) = \Theta(a)$. To complete our construction, we cut off a small part from the top of each edge gadget with area more than $b$ (however, we have to keep the $a$ areas intact), as illustrated in Figure~\ref{fig:reduction_maxkvertexcover}, to ensure all edge gadgets have the same area. The construction can clearly be done in polynomial time.

     We give the watchman a budget of $B = 2l + 2k$. If the watchman starts at (or near) any edge gadget, he can only see an area of $b + 2a$. If the watchman starts in the middle of any node-edge corridor, he must travel down to the node gadget, the bottom corridor and then other node gadgets to see more of $P$. The watchman does not see any more area of $P$ when traveling down from a node-edge corridor to the node gadget; as a result, we can impose that the watchman must start somewhere in the region of the bottom corridor and the node gadgets.

     If an optimal solution of the \textsc{Max $k$-Vertex Cover} instance covers $m$ edges, an optimal solution of the BWRP instance travels the bottom corridor, selectively goes up some node gadgets to see an area of $3ka + mb$, minus the area of the crossings between corridors. Suppose we have a polynomial-time algorithm that approximates the BWRP to a factor of $(1 - \varepsilon)$, run that algorithm on $P$ and let $m'$ be the number of edge gadgets the approximate solution sees. Also, let $A_1$ be the total area of crossings of corridors visible to node gadgets selected by the approximate solution, and $A_2$ be that by the optimal solution, clearly $A_1 - (1 - \varepsilon)A_2 \ge~-\varepsilon b$. Hence, 
     \begin{align*}
         3ka + m'b - A_1 &\ge (1 - \varepsilon)(3ka + mb - A_2)\\
         \Leftrightarrow 3ka + m'b &\ge (1 - \varepsilon)(3ka + mb) - \varepsilon b\\
         \Leftrightarrow 3\varepsilon k a + m'b &\ge (1-\varepsilon)mb - \varepsilon b,
    \end{align*}
    moreover, since $m \ge m' \ge k \ge 1$ and $b = \Theta(a)$
    \begin{align*}
          O(1)\varepsilon mb + m'b &\ge \left(1 - 2\varepsilon\right)mb\\
         \Leftrightarrow m' &\ge (1 - O(1)\varepsilon)m.
     \end{align*}
    Thus, the nodes selected by the aforementioned approximate solution to the BWRP instance is a $(1 - \varepsilon')$-approximation, for some $\varepsilon' > 0$ arbitrarily small, to the \textsc{Max $k$-Vertex Cover} instance. This leads to a contradiction, and our proof of APX-hardness of the BWRP in a polygon with holes is complete.
\end{proof}

    \subsection{Approximation algorithm for the BWRP in a polygon with holes}
    We decompose $P$ into small convex cells and obtain the set of candidates $S_{\delta, B}$ as in the case with the BWRP in a simple polygon.

    \begin{theorem}
        There exists a route $\gamma'$ whose vertices are a subset of $S_{\delta, B}$ such that $V(\gamma) \subseteq V(\gamma')$ and $|\gamma'| \le (1+\varepsilon)B$.
    \end{theorem}

    \begin{proof}
        Consider an edge $e$ of $\gamma$ whose endpoints lie in cells $\sigma_i$ and $\sigma_j$. We append $\partial\sigma_i$ and $\partial\sigma_j$ to $\gamma$. If the endpoints of $e$ are not vertices of $\sigma_i$ and $\sigma_j$, we replace $e$ with a set of edges whose endpoints are candidate points. The procedure can be described with a physical analog: we slide an elastic string between the two intersection points of $e$ with $\sigma_i$ and $\sigma_j$ towards the exterior of $\gamma$ until the two endpoints coincide with vertices of $\sigma_i$ and $\sigma_j$, the string is pulled taut and never passes through a hole; see Figure~\ref{fig:snapping_in_with_holes}. The result is a geodesic path $e'$ between two vertices of $\sigma_i$ and $\sigma_j$ that is no longer than $|e| + 2\sqrt{2}\delta$.

        Repeating the process for every edge of $\gamma$, we obtain $\gamma'$. If a point $x$ seen by $\gamma$ is inside of $P_{\gamma'}$, the extended line of vision between $x$ and $\gamma$ must intersect with $\gamma'$ since there is no hole between $\gamma$ and $\gamma'$. If $x$ is outside of $P_{\gamma'}$, then the line of vision between $x$ and $\gamma$ must intersect $\gamma'$ due to the Jordan Curve Theorem. Thus, $\gamma'$ sees everything that $\gamma$ sees, moreover $\gamma'$ passes through $s$, a vertex in the decomposition. Since $\gamma$ has $O(n^2)$ vertices \cite{mitchell2013approximating}, for an appropriate choice of $\delta = O\left(\frac{\varepsilon B}{n^2}\right)$ we have $|\gamma'| \le (1+\varepsilon)B$.
    \end{proof}
   
    \begin{figure}[h]
        \centering
        \includegraphics[width=0.7\textwidth]{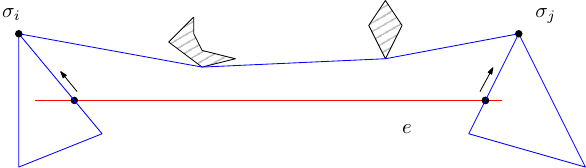}
        \caption{Replacing each edge (red) with the perimeters of the two cells containing its endpoints and a geodesic path of the same homotopy type (blue).}
        \label{fig:snapping_in_with_holes}
    \end{figure}

    We apply a known result for the \textsc{Submodular Orienteering} problem~\cite{chekuri2005recursive}: Given a weighted directed graph $G$, two nodes $s$ and $t$ (which need not be distinct), a budget $B > 0$, and a monotone submodular reward function defined on the nodes, find an $s$-$t$ walk that maximizes the reward, under the constraint that the length of the walk is no greater than~$B$.
    
   Let $G_1$ be the visibility graph on the candidates set with Euclidean edge weights. Let $G_2$ be the line graph of $G_1$: nodes of $G_2$ correspond to edges of $G_1$, and two nodes in $G_2$ are adjacent if their respective edges in $G_1$ are incident. 
   The weight of an edge of $G_2$ is the sum of the weights of the two edges in $G_1$ corresponding to its endpoints, divided by two, thus a closed walk of length $B$ in $G_1$ corresponds to a closed walk of length $B$ in $G_2$ and vice versa.  We apply the approximation algorithm from~\cite{chekuri2005recursive} on $G_2$ to compute a closed walk from any node in $G_2$ corresponding to an edge incident with $s$, with the area of visibility as the reward function and budget $(1+\varepsilon)B$. The reason for using the line graph $G_2$ is that, in the \textsc{Submodular Orienteering} problem, rewards are associated with nodes, while in the context of the BWRP, rewards are accumulated when traversing edges of the visibility graph $G_1$. We obtain the following:

   %% xxx state what beta is in the Theorem?? or remove it altogether since it is a constant?? The approx factor is log(n), not log(OPT), right??  Fix mention in intro as well!  Does log(B) make any sense in the running time?
\old{reminder: SoCG Reviewer 2:    line 47:  log(OPT) does not make sense in this form, as it is not
scale-invariant.    The approximation factor doesn't improve when you
scale the input by, say, 0.001.}
    
    \begin{theorem}
        Given a polygon $P$ with holes with $n$ vertices, let $\beta \ge 2$ be any constant of choice and $OPT$ be the maximum area that a route of length $B$ can see. The BWRP has a dual approximation algorithm that computes a tour of length at most $(1+\varepsilon)B$ that sees an area of at least $\Omega\left(\frac{OPT\log \beta}{\log n}\right)$, with running time  $\left(\frac{n}{\varepsilon}\log B\right)^{O\left(\beta\log \frac{n}{\varepsilon}/\log \beta\right)}$.
    \end{theorem}
    %xxx: changed to log n factor to address reviewer 2
    \begin{proof}
        Let $k$ be the number of edges in the walk $\gamma'$ returned by the recursive greedy algorithm in \cite{chekuri2005recursive}. Then $|V(\gamma')| = \Omega\left(\frac{OPT\log \beta}{\ceil{1 + \log k}}\right)$. Since $G_2$ has $O(n^4)$ nodes, it follows that $k = O(n^4)$ and $|V(\gamma')| = \Omega\left(\frac{OPT\log \beta}{\log n}\right)$.
    \end{proof}

    %% Moved section to appendix:

\section{Optimal Visibility-based Search for a Randomly Distributed Target}
Our results can be applied to solve two problems of searching a randomly distributed static target in a simple polygon $P$: Given a  prior distribution of the target's location in $P$, (1) compute a route that achieves a given detection probability within the minimum amount of time, where the target is detected if the watchman can see it; and (2) (dual to (1)) for a given time budget $T$, compute a search route maximizing the probability of detecting the target by time $T$. Denote by $\mu(.)$ the probability measure on all subsets of $P$; $\mu(P_1)$ is the probability measure of $P_1\subseteq P$, i.e. the probability that the target is in $P_1$, then
\begin{itemize}
    \item $0 \le \mu(.) \le 1, \mu(\varnothing) = 0, \mu(P) = 1$,
    \item $\mu(P_1\cup P_2) = \mu(P_1) + \mu(P_2)$ if $P_1\cap P_2 = \varnothing$.
\end{itemize}
We assume that we have access to $\mu(.)$ via an oracle: Given a triangular region in $P$, the oracle returns its probability measure in $O(1)$ time. Thus, for a point or a segment, the probability measure of its visibility region can be computed in $O(n)$ time. Furthermore, if the watchman has constant speed, a time constraint/objective is equivalent to that of length. An optimal search route for each problem can be computed using the algorithms given with probability measure instead of area. 
% \linh{Maybe now that we have room once again this section doesn't have to be in the appendix anymore}

\bibliographystyle{plainurl}% the mandatory bibstyle
\bibliography{refs.bib}

\end{document}